\pdfoutput=1
\documentclass[a4paper,11pt]{article}
\usepackage{amsmath}
\usepackage{amssymb}
\usepackage{hyperref}
\usepackage{cite}
\usepackage{graphicx}
\usepackage{subfig}
\usepackage{authblk}

\newcommand{\Tr}{\mathrm{Tr}}
\newcommand{\tr}{\Tr}

\newcommand{\ket}[1]{\ensuremath{|#1\rangle}}
\newcommand{\bra}[1]{\ensuremath{\langle#1|}}
\newcommand{\ketbra}[2]{\ensuremath{\ket{#1}\bra{#2}}}
\newcommand{\braket}[2]{\ensuremath{\langle{#1}|{#2}\rangle}}
\newcommand{\VV}{\mathcal{V}}

\newcommand{\1}{{\rm 1\hspace{-0.9mm}l}}
\newcommand{\Id}{\1}

\newcommand{\ie}{{\emph{i.e.\/}}}

\newcommand{\E}{\mathbb{E}}
\newcommand{\F}{\mathcal{F}}
\newcommand{\N}{\mathcal{N}}
\newcommand{\Z}{\mathbb{Z}}

\newcommand{\R}{\mathbb{R}}

\newcommand{\etal}{\emph{et al.}}

\newcommand{\XX}{\mathcal{X}}
\newcommand{\YY}{\mathcal{Y}}
\newcommand{\ZZ}{\mathcal{Z}}

\newcommand{\ImConj}{\mathrm{Im}(\Phi^\dagger)}
\newcommand{\KerPhi}{\mathrm{Ker}(\Phi)}
\newcommand{\Ker}{\mathrm{Ker}}
\newcommand{\Image}{\mathrm{Im}}

\newtheorem{theorem}{Theorem}
\newtheorem{lemma}{Lemma}
\newtheorem{definition}{Definition}

\newenvironment{proof}[1][Proof]{\noindent\textbf{#1.} }{\hfill 
\rule{0.5em}{0.5em}}

\title{Central limit theorem for reducible and irreducible open quantum walks}

\author{Przemys{\l}aw Sadowski\thanks{psadowski@iitis.pl}\;}
\author{{\L}ukasz Pawela\thanks{lpawela@iitis.pl}}
\affil{Institute of Theoretical and Applied Informatics, Polish Academy of 
Sciences, Ba{\l}tycka 5, 44-100 Gliwice, Poland}

\begin{document}
\maketitle

\begin{abstract}
In this work we aim at proving central limit theorems for open quantum walks on
$\mathbb{Z}^d$. We study the case when there are various classes of vertices in
the network. Furthermore, we investigate two ways of distributing the vertex
classes in the network. First we assign the classes in a regular pattern.
Secondly, we assign each vertex a random class with a uniform distribution. For
each way of distributing vertex classes, we obtain an appropriate central limit
theorem, illustrated by numerical examples. These theorems may have application
in the study of complex systems in quantum biology and dissipative quantum
computation.
\end{abstract}

\section{Introduction}\label{sec:introduction}
In a series of recent papers
\cite{attal12open,attal2012open,sinayskiy2013open,sweke2013dissipative,sinayskiy2012properties,pawela2014generalized}
various aspects of open quantum walks have been discussed. This is a novel and 
very promising approach to the quantum walks.
Quantum walks have long been studied~\cite{Reitzner, 
Ambainis2008quantum, Ampadu, Spatiotemporal1201.4839, percolation} and have 
numerous applications, such as: search algorithms \cite{grover1997search, 
kempe2003walksearch, portugal2013quantum, childs2004spatialsearch, 
sadowski04efficient}, quantum agents~\cite{miszczak2014magnus} and quantum 
games~\cite{flitney2002introduction,piotrowski2003invitation,pawela2013cooperative,pawela2013quantum}.
Open walks generalize this well studied model and in particular allow one to
incorporate decoherence which is an always present factor when considering
quantum systems.
The importance of managing decoherence has motivated the study of this problem 
in numerous fields, such as quantum 
control~\cite{dahleh1990optimal,viola1999universal,viola2003robust,james2004risk,d2006quantum,dong2009sliding,pawela2014quantum,pawela2014quantum2,gawron2014decoherence,pawela2013various},
quantum 
games~\cite{du2002experimental,flitney2005quantum,flitney2007multiplayer,pawela2013enhancing,gawron2014relativistic}and
 quantum walks~\cite{grover_fail, Franco1303.5319, Kendon, kendon0209005v3, 
ampadu0localization, Chandrashekar_few_ie2d }.

In this work we analyze the asymptotic behavior of open quantum walks.
Especially we consider the possibility to determine the time limit properties
of walks with non-homogeneous structure. The theorems for the homogeneous case
are proven~\cite{attal2012central}. In this work we consider two different
approaches: the possibility to reduce the walk to the homogeneous one and
provide walk's asymptotic properties as it is. In the first case, we construct
a set of rules and methods that allows to determine when it is possible to
reduce a walk. In the second case, we state a new central limit theorem that
allows us to derive asymptotic distribution under certain conditions. We
illustrate this approaches with appropriate numerical examples.

\section{Preliminaries}\label{sec:oqw}
\subsection{Quantum states and channels}
\begin{definition}
We call an operator $\rho \in L(\XX)$, for some Hilbert space $\XX$, a density
operator iff $\rho \geq 0$ and $\tr \rho = 1$. We denote the set of all density
operators on $\XX$ by $\Omega(\XX)$.
\end{definition}

\begin{definition}
A superoperator $\Phi$ is a linear mapping acting on linear operators $L(\XX)$
on a finite dimensional Hilbert space $\XX$ and transforming them into
operators on another finite dimensional Hilbert space $\YY$ i. e.
\begin{equation}
	\Phi: L(\XX) \rightarrow L(\YY).
\end{equation}
\end{definition}

\begin{definition}
Given superoperators 
\begin{equation}
\Phi_1: L(\XX_1) \rightarrow L(\YY_1), \ \Phi_2: L(\XX_2) 
\rightarrow 
L(\YY_2),
\end{equation}
we define the product superoperator
\begin{equation}
	\Phi_1 \otimes \Phi_2: L(\XX_1 \otimes\XX_2) 
	\rightarrow L(\YY_1\otimes \YY_2),
\end{equation}
to be the unique linear mapping that satisfies:
\begin{equation}
(\Phi_1 \otimes \Phi_2)(A_2 \otimes  A_2) = 
\Phi_1(A_1) \otimes \Phi_2(A_2),
\end{equation}
for all operators $A_1\in L(\XX_1), A_2 \in L(\XX_2)$. The extension for 
operators not in the tensor product form follows from linearity.
\end{definition}

\begin{definition}\label{def:channel}
A quantum channel is a superoperator $\Phi: L(\XX) \rightarrow L(\YY)$ that
satisfies the following restrictions:
\begin{enumerate}
\item $\Phi$ is trace-preserving, i.e. $\forall {A \in L(\XX)} \ 
\tr(\Phi(A))=\tr(A)$,

\item \label{item:CP}$\Phi$ is completely positive, that is for every 
finite-dimensional Hilbert space $\ZZ$ the product of $\Phi$ and an identity 
mapping on $L(\ZZ)$ is a non-negativity preserving operation, i.e.
\begin{equation}
\forall {\ZZ} \ \forall {A \in L(\XX \otimes \ZZ)}, \ {A \geq 0} \ 
(\Phi\otimes \1_{L(\ZZ)})(A) \geq 0.
\end{equation}
\end{enumerate}
\end{definition}
Note that quantum channels map density operators to density operators.

\begin{definition}
The Kraus representation of a quantum channel $\Phi: L(\XX) \rightarrow L(\YY)$
is given by a set of operators  $K_i \in L(\XX, \YY)$. The action of the 
superoperator  $\Phi$ on $A \in L(\XX)$ is given by:
\begin{equation}
\Phi(A)=\sum_i K_i A K_i^\dagger,
\end{equation}
with the restriction that
\begin{equation}
	\sum_i K_i^\dagger K_i=\1_{\XX}.
\end{equation}
\end{definition}

\begin{definition}
Given a superoperator $\Phi: L(\XX) \rightarrow L(\YY)$, for every operator $A
\in L(\XX), B \in L(\YY)$ we define the conjugate superoperator $\Phi^\dagger:
L(\YY) \rightarrow L(\XX)$ as the mapping satisfying
\begin{equation}
\forall A \in L(\XX) \; \forall B \in L(\YY) \quad \Tr (\Phi(A)B) = \Tr (A 
\Phi^\dagger(B)). \label{eq:conjugate-channel}
\end{equation}
\end{definition}
Note, that the conjugate to a completely positive superoperator is completely 
positive, but is not necessarily trace-preserving.

\subsection{Open quantum walks}

The model of the open quantum walk was introduced by Attal \etal\
\cite{attal12open} (see also \cite{sinayskiy12open}). To introduce the open
quantum walk (OQW) model, we consider a random walk on a graph with the set
of vertices $V$ and directed edges $\{(i, j): \; i, j \in V\}$. The
dynamics on the graph is described in the space of states $\VV =
\mathbb{C}^V$ with an orthonormal basis $\{ \ket{i} \}_{i \in V}$. We model
an internal degree of freedom of the walker by attaching a Hilbert space
$\XX$ to each vertex of the graph. Thus, the state of the quantum walker is
described by an element of the space $\Omega(\XX \otimes \VV)$.

To describe the dynamics of the quantum walk, for each directed edge $(i, j)$ 
we introduce a set of  operators $\{K_{ijk} \in L(\XX)\}$. These 
operators 
describe the change in the internal degree of freedom of the walker due to the 
transition from vertex $j$ to vertex $i$. Choosing the operators $K_{ijk}$ such 
that
\begin{equation}
\sum_{ik} K_{ijk}^\dagger K_{ijk} = \1_{\XX},
\end{equation}
we get a Kraus representation of a quantum channel for each vertex $j \in V$ of 
the graph. As the operators $K_{ijk}$ act only on $\XX$, we introduce the 
operators $M_{ijk} \in L(\XX \otimes \VV)$
\begin{equation}
M_{ijk} = K_{ijk} \otimes \ketbra{i}{j},\label{eq:kraus-channel}
\end{equation}
where$\ket{i}, \ket{j} \in \VV$ which perform the transition from vertex $j$ to
vertex $i$ and internal state evolution. It is straightforward to check that 
$\sum_{ijk} M_{ijk}^\dagger
M_{ijk} = \1_{\XX \otimes \VV}$.
\begin{definition}
A discrete-time open quantum walk is given by a quantum channel $\Phi: L(\XX
\otimes \VV) \rightarrow L(\XX \otimes \VV)$ with the Kraus representation
\begin{equation}
\forall A \in L(\XX \otimes \VV) \quad \Phi(A) = \sum_{ijk} M_{ijk} A 
M_{ijk}^\dagger,
\end{equation}
where operators $M_{ijk} \in L(\XX \otimes \VV)$ are defined in
Eq.~\eqref{eq:kraus-channel}.
\end{definition}

\subsection{Asymptotic behavior of open quantum walks}

Recently Attal \etal \cite{attal2012central} provided a description of
asymptotic behavior of open quantum walks in the case when the behavior of
every vertex is the same i. e. all vertices belong to one class. We call such
networks homogeneous.

In order to describe asymptotic properties of an open quantum walk we will 
use the notion of quantum trajectory process associated with this open quantum 
walk.
\begin{definition}
We define the quantum trajectory process as a classical Markov chain assigned
to an open quantum walk constructed as a simulation of the walk with
measurement at each step. The initial state is  $(\rho_0, X_0)\in \Omega(\XX)
\times \mathbb{Z}^d$ with probability 1. The state $(\rho_n, X_n)$ at step $n$
evolves into one of the $2d$ states corresponding to possible directions
$\Delta_j$, $j=\pm1,\ldots,\pm d$:
\begin{equation}
\left(\frac{1}{p_j}K_j \rho K_j^\dagger, X_n+\Delta_j\right),
\end{equation}
with probability $p_j=\Tr(K_j\rho K_j^\dagger)$. We also separately define a 
Markov chain $(\rho, \Delta X)$ and a transition
operator associated with this trajectory process
\begin{equation}\label{eq:definition:P}
P[(\rho, \Delta_i), (\rho', \Delta_j)]= 
\left\{\begin{array}{ll}
\Tr(K_j\rho K_j^\dagger) 
& \mathrm{if} \rho'=\frac{K_j \rho K_j^\dagger}{Tr(K_j \rho 
K_j^\dagger)}, \\
0 & \mathrm{else}.
\end{array}\right.
\end{equation}
\end{definition}

We define an auxiliary channel $\Phi: L(\XX) \rightarrow L(\XX)$ that mimics
the behavior of the walk when all the internal states are the same as
\begin{equation}
\Phi(\rho)=\sum_{j=1}^{2d} K_j \rho K_j^\dagger.\label{eq:aux-channel}
\end{equation}
We assume that the channel has a unique invariant state $\rho_\infty 
\in 
\Omega(\XX)$. 
Additionally we define a vector that approximates the estimated asymptotic 
transition for the channel $\Phi$:
\begin{equation}
\ket{m}=\sum_{j=1}^{2d}\Tr(K_j\rho_{\infty}K_j^\dagger)\ket{j},\label{eq:aux-vector}
\end{equation}
where $\ket{j} \in \R^d$ and for $j>d$ we put $\ket{j}=-\ket{j-d}$.

Let us recall the theorem by Attal \etal~\cite{attal2012central}. First we 
recall a simple lemma:
\begin{lemma}
For every $\ket{l} \in \R^d$ and channel $\Phi$ with associated Kraus operators 
$\{ K_1, \ldots, K_{2d}\}$, the equation
\begin{equation}
(L_l - \Phi^\dagger(L_l)) = \sum_{i=1}^{2d} K_i^\dagger K_i \ketbra{i}{l} - 
\ketbra{m}{l} \1
\end{equation}
admits a solution
\end{lemma}
We will write $L_i$ instead of $L_l$ for $\ket{l} = \ket{i}$.
Now, we can state the theorem
\begin{theorem}\label{th:original}
Consider a open quantum walk on $\mathbb{Z}^d$ associated with transition
operators $\{K_1, \ldots, K_{2d}\}$. We assume that a channel  $\Phi$ admits a
unique invariant state. Let $(\rho_n, X_n)_{n\ge 0}$ be the quantum trajectory
process associated with this open quantum walk, then
\begin{equation}
\lim\limits_{n \rightarrow \infty} \frac{\E(\ket{X_n})}{n} = \ket{m},
\end{equation}
and probability distribution of normalized random variable $X_n$
\begin{equation}
\frac{\ket{X_n}-n \ket{m}}{\sqrt{n}},
\end{equation}
converges in law to the Gaussian distribution $\mathcal{N}(0,C)$ in 
${\mathbb{R}}^d$, with the covariance matrix
\begin{equation}\
\begin{split}
C_{ij}=&\delta_{ij} \left( \tr(K_i \rho_\infty K_i^\dagger) + \tr(K_{i+d} 
\rho_\infty K_{i+d}^\dagger) \right) - m_i m_j + \\
& + \left( \Tr(K_i \rho_\infty K_i^\dagger L_j) + \Tr(K_j \rho_\infty 
K_j^\dagger L_i) \right. \\
& \left. - \Tr(K_{i+d} \rho_\infty K_{i+d}^\dagger L_j) - \Tr(K_{j+d} 
\rho_\infty K_{j+d}^\dagger L_i) \right) \\
& \left( - m_i \tr(\rho_\infty L_j) - m_j \tr(\rho_\infty L_i) \right).
\end{split}
\end{equation}
\end{theorem}

\section{Results}

We are mainly interested in open quantum walks that are defined on networks
with many classes of vertices. In this paper we assume that there is a finite
number of vertex classes $\Gamma=\{C_1, \ldots, C_n\}$. The transitions in each
vertex is given by Kraus operators defined for each class separately
$\{K_1^{c(X)}, \ldots, K_{2d}^{c(X)}\}_{X\in \Z^d} \subset L(\XX)$, where
$c(X)\in\Gamma$ is class of the vertex $X\in\Z^d$, Such that $\sum_j
\left(K_j^C \right)^\dagger K_j^C=\1_{\XX}$. We define a transition operator of
the Markov chain as in Eq.~\eqref{eq:definition:P}:
\begin{equation}\label{eq:pc}
P_C[(\rho, \Delta_i), (\rho', \Delta_j)]= 
\left\{\begin{array}{ll}
\Tr\left(K_j^C\rho \left( K_j^C \right)^\dagger\right) 
& \mathrm{if} \rho'=\frac{K_j^C\rho \left( K_j^C 
\right)^\dagger}{\Tr\left(K_j^C 
\rho 
\left( K_j^C \right)^\dagger\right)}, \\
0 & \mathrm{else}.
\end{array}\right.
\end{equation}
Next, we define a channel $\Phi^C$ for each class $C$ as in 
Eq.~(\ref{eq:aux-channel})
\begin{equation}
\Phi^C(\rho)=\sum_{j=-d, j\ne 0}^{d} K_j^C \rho (K_j^C)^\dagger.\label{eq:phic}
\end{equation}
Again, we assume that $\Phi^C$ has a unique invariant state $\rho_\infty^C \in 
\Omega(\XX)$. Additionally for each class $C$ we define a vector as in 
Eq.~(\ref{eq:aux-vector})
\begin{equation}
\ket{m_C}=\sum_{j=1}^{2d} 
\Tr(K_j^C\rho_{\infty}^C(K_j^C)^\dagger)\ket{j}\label{eq:mc},
\end{equation}
where $\ket{j} \in \R^d$ and for $j>d$ we put $\ket{j}=-\ket{j-d}$.

In order to provide a description of distribution evolution of open quantum 
walks on non-homogeneous networks we analyze two cases. First in 
Section~\ref{sec:reducible}, we model a walk with vertices defined in such a 
way that it is possible to reduce the network to the homogeneous case. Secondly 
in Section~\ref{sec:irreducible}, we study a network which is irreducible in 
the above sense but satisfies some basic properties that allow us to develop 
other techniques.

\subsection{Reducible open quantum walks}\label{sec:reducible}

Let us consider an open quantum walk with several classes of vertices. We aim
to analyze the possibility to construct a new walk that behaves the same way in 
the asymptotic limit.
\begin{definition}
We call an open quantum walk reducible if there is a class $A$ that for some
integer $l$ each $l$-step path from a vertex of type $A$ always  leads to a
vertex of type $A$.

\end{definition}
When considering a reducible OQW we can consider these paths as edges and
reduce the network to the homogeneous case.
\begin{definition}\label{def:abstract-class}
For a reducible quantum walk with $N$ possible paths we construct a new set of
Kraus operators $\{K_1^R, .., K_N^R\} \subset L(\XX)$ such that each operator
is a composition of all the operators corresponding to the consecutive steps
composing one of the paths from vertex $A$ to another vertex $A$, i. e. for a
path $q$ consisting of vertices $X_1, \dots, X_l$ and direction changes
$\Delta_1, \ldots, \Delta_l$ the corresponding operator is
\begin{equation}
K_q^R = K_{\Delta_l}^{C(X_l)}\cdot\ldots\cdot K_{\Delta_1}^{C(X_1)}, 
q=1,\ldots,N.
\end{equation}
We call the OQW based on these operators a reduced open quantum walk.
\end{definition}

The simplest example of a reducible open quantum walk is a walk on
$\mathbb{Z}^2$ presented in Fig.~\ref{fig:2d-reducible-network}. Starting in a
vertex of class $A$, after two steps we always end up in a vertex of class $A$.
We use that property to construct a new walk with only one vertex type and
exactly the same asymptotic behavior. In 
Fig.~\ref{fig:2d-reducible-network-big} we present a more complex example of a
network with these properties.

\begin{figure*}[!ht]
\centering\includegraphics{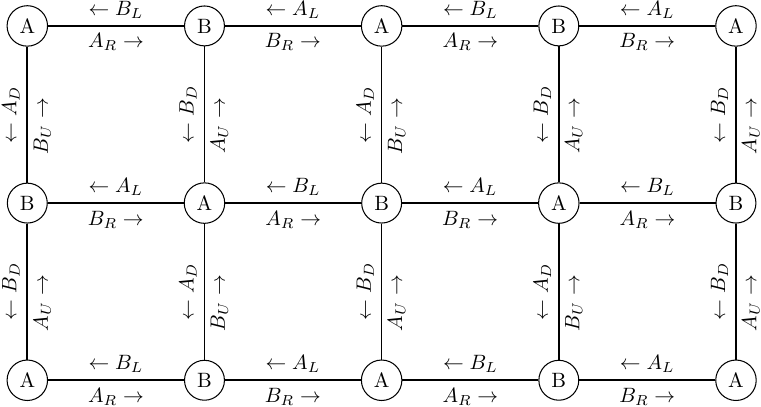}
\caption{An example of a 2D reducible OQW. The operators are defined in the 
text. The dashed lines show possible paths from one vertex of type $A$ to 
another 
vertex of this type.}\label{fig:2d-reducible-network}
\end{figure*}
\begin{figure*}[!ht]
\centering\includegraphics{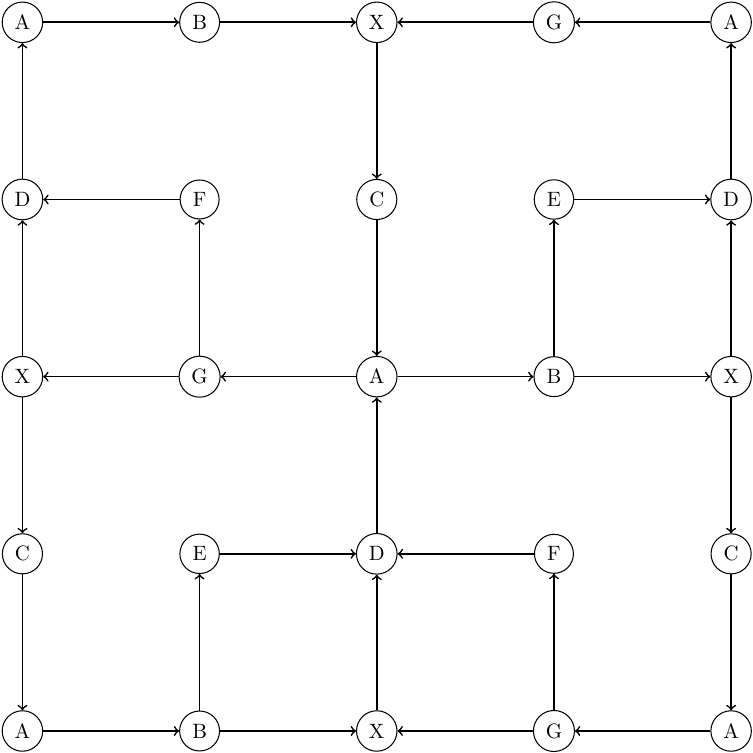}
\caption{An example of a 2D reducible OQW. The arrows show possible 
transitions. Each path from one vertex of type $A$ leads to another vertex of 
this 
type with exactly 4 steps.} 
\label{fig:2d-reducible-network-big}
\end{figure*}

\subsubsection{Central limit theorem and its proof}

\begin{theorem}\label{th:reducible}
Consider a reducible open quantum walk on $\mathbb{Z}^d$. By $P$ we denote the
abstract class of vertices constructed as described in
Definition~\ref{def:abstract-class}. We assume that a channel constructed with
these paths $\Phi^P$ has a unique invariant state $\rho_\infty \in \Omega(\XX)$ 
with 
average transition vector $\ket{m_P}$. Let $(\rho_n, X_n)_{n\ge 0}$ be the 
quantum trajectory process associated to this open quantum walk, then
\begin{equation}
\lim_{n\rightarrow\infty} \frac{\E( \ket{X_n})}{n}=\ket{m_P},
\end{equation}
and probability distribution of normalized random variable $X_n$
\begin{equation}
\frac{\ket{X_n}-n \ket{m_P}}{\sqrt{n}},
\end{equation}
converges in law to the Gaussian distribution in ${\mathbb{R}}^d$.
\end{theorem}

\begin{proof}
We apply the Theorem~\ref{th:original} to the reduced OQW as in Def.
\ref{def:abstract-class}. As all the path's lengths are equal and describe all
possible paths starting from a vertex of type $A$ we have that $\sum_{q=1}^N
K_q^{R\dagger} K_q^R = \sum_{q=1}^{N} \left( K_1^q\ldots K_l^q \right)^\dagger
K_1^q \ldots K_l^q = \1_{\XX}$. Thus the new walk satisfies assumptions of the
Theorem~\ref{th:original}. One step of this walk corresponds exactly to $l$
steps of the original walk.  The one-to-one correspondence assures that the
asymptotic behavior is the same.
\end{proof}

\subsubsection{Example}\label{sec:example-reducible}
We show the application of Theorem~\ref{th:reducible} by considering a walk on
a network presented in the Fig. \ref{fig:2d-reducible-network}. The Kraus
operators for vertices of type $A$ are defined as follows:
\begin{equation}
\begin{split}
A_U(X) = & \alpha \ket{0}\bra{0}X\ket{0}\bra{0} + (1 - \alpha) \ket{1}\bra{0} 
X \ket{0}\bra{1}, \\
A_R(X) = & \frac12 \ket{1}\bra{1}X\ket{1}\bra{1} + \frac12 \ket{3}\bra{1} X 
\ket{1}\bra{3} \\
A_D(X) = & \alpha \ket{3}\bra{2}X\ket{3}\bra{2} + (1 - \alpha) \ket{2}\bra{2} 
X \ket{2}\bra{2}, \\
A_L(X) = & \frac12 \ket{3}\bra{3}X\ket{3}\bra{3} + \frac12 \ket{0}\bra{3} 
X \ket{3}\bra{0}.
\end{split}
\end{equation}
The operators for vertices of type $B$ are:
\begin{equation}
\begin{split}
B_U(X) = & \alpha \ket{1}\bra{0}X\ket{0}\bra{1} + (1 - \alpha) \ket{3}\bra{0} 
X \ket{0}\bra{3}, \\
B_R(X) = & \frac12 \ket{0}\bra{1}X\ket{1}\bra{0} + \frac12 \ket{2}\bra{1} X 
\ket{1}\bra{2}, \\
B_D(X) = & \alpha \ket{1}\bra{2}X\ket{2}\bra{1} + (1 - \alpha) \ket{3}\bra{2} 
X \ket{2}\bra{3}, \\
B_L(X) = & \frac12 \ket{0}\bra{3}X\ket{3}\bra{0} + \frac12 \ket{2}\bra{3} 
X \ket{3}\bra{2}.
\end{split}
\end{equation}
In our example we set $\alpha=0.81$. The behavior of this particular walk is 
presented in Fig.~\ref{fig:2d-reducible-walk}. As expected, after a 
sufficiently large 
number of steps, the distribution is Gaussian and moves towards the left and 
down.
\begin{figure}[!h]
\centering
\subfloat[\label{fig:reducible-a}]{\includegraphics[width=0.49\textwidth]{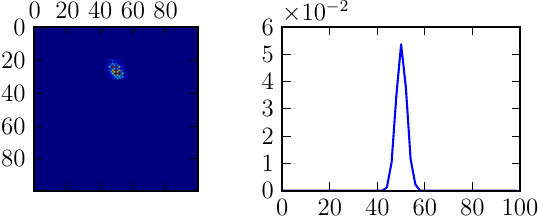}}
\subfloat[\label{fig:reducible-b}]{\includegraphics[width=0.49\textwidth]{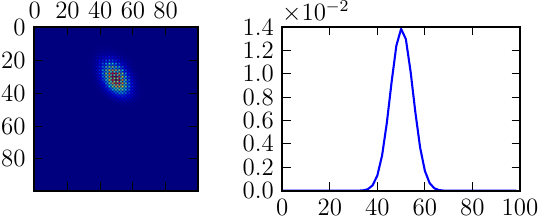}}\\
\subfloat[\label{fig:reducible-c}]{\includegraphics[width=0.49\textwidth]{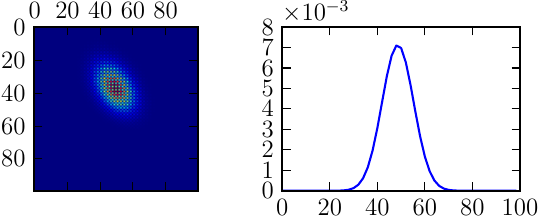}}
\subfloat[\label{fig:reducible-d}]{\includegraphics[width=0.49\textwidth]{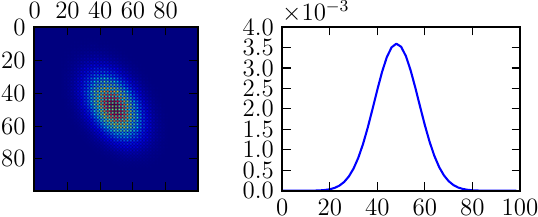}}
\caption{An example of a reducible OQW on a 2D lattice shown in
Fig.~\ref{fig:2d-reducible-network}. The plots show the distribution of the
walks for various time steps and a cross section through the center of the
distribution. Panel~\protect\subref{fig:reducible-a}  $n=10$,
panel~\protect\subref{fig:reducible-b}  $n=50$,
panel~\protect\subref{fig:reducible-c}  $n=100$,
panel~\protect\subref{fig:reducible-d}  $n=200$.}\label{fig:2d-reducible-walk}
\end{figure}

\subsection{Irreducible OQWs}\label{sec:irreducible}
The assumptions introduced in Theorem~\ref{th:reducible}
allow us to analyze some non-homogeneous OQW, but the class of such walks is 
still very limited. In this section we aim to provide a way to determine 
asymptotic behavior of less restricted family of OQWs.

\subsubsection{Theorem and proof}\label{sec:irreducible-theorem}
Let's consider an OQW on a network composed with several types of vertices on an
infinite lattice. The main assumption of the following theorem is that the 
distribution of vertex classes is regular over the lattice \ie~density of every 
vertex class $C\in\Gamma$ is transition invariant.

\begin{definition}\label{def:regular-network}
A regular network is a network where each vertex's class is assigned
randomly at each step with transition invariant probability distribution 
$\{p_C\}_C$.
\end{definition}

\begin{theorem}\label{th:irreducible}
Given an open quantum walk on $\mathbb{Z}^d$ with vertex classes 
$c(X)\in\Gamma$ for $X 
\in \mathbb{Z}^d$ and associated transition operators $\{K_1^{c(X)}, \ldots, 
K_{2d}^{c(X)}\}_{X\in \Z^d} \subset L(\XX)$ we construct for each class of 
vertices $C\in\Gamma$ a quantum channel $\Phi^C$ as in Eq.~\eqref{eq:phic} with 
a unique 
invariant state $\rho_\infty^C \in \Omega(\XX)$ and an average position vector 
$\ket{m} = \sum_{C\in\Gamma} p_C \ket{m_C}$, where $\ket{m_C}$ is 
obtained from Eq.~\eqref{eq:mc} and $p_C$ from Def. \ref{def:regular-network}. 
Let $(\rho_n, X_n)_{n\ge 0}$ be the quantum 
trajectory process associated with this open quantum walk, then
\begin{equation}
\lim_{n\rightarrow\infty} \frac{\mathbb{E}(\ket{X_n})}{n}= \ket{m},
\end{equation}
and probability distribution of normalized random variable $X_n$
\begin{equation}
\frac{\ket{X_n}- n \ket{m}}{\sqrt{n}},
\end{equation}
converges in law to the Gaussian distribution.
\end{theorem}
Before we prove Theorem~\ref{th:irreducible}, let us introduce three technical 
lemmas.

\begin{lemma}\label{lem:ker-im-sum}
For every superoperator $\Phi:L(\XX)\rightarrow L(\XX)$ the space 
$L(\XX) = \mathrm{Ker}(\Phi) \oplus \mathrm{Im}(\Phi^\dagger)$.
\end{lemma}
\begin{proof}
First we show that if $A\perp \ImConj$ then $A\in\KerPhi$ for $A\in L(\XX)$.
Let us assume $A\perp \ImConj$. Then for every $B\in L(\XX)$ it holds that 
$\tr(A\Phi^\dagger (B))=0$. Then $\tr(\Phi(A)B)=0$. Thus $\Phi(A)=0$ and 
$A\in\KerPhi$.

Now we show that if $A\in\KerPhi$ then $A\perp\ImConj$. We assume $\Phi(A)=0$.
Then for any chosen $B\in L(\XX)$ it holds that $\tr(\Phi(A)B)=0$, thus 
$\tr(A\Phi^\dagger(B)=0)$, hence $A\perp\ImConj$.
\end{proof}

\begin{lemma}\label{lem:Lexists}
Given a channel $\Phi^C$ corresponding to vertex class $C$ with associated
Kraus operators $\{K_1^C, \ldots, K_{2d}^C\} \subset L(\XX)$ which has a unique
invariant state $\rho_\infty$, for every $\ket{l}\in\mathbb{R}^d$ there exists 
$L_l^C
\in L(\XX)$ such that
\begin{equation}
\left(\Id_{L(\XX)} -\left( \Phi^C \right)^\dagger \right)\left(L_l^C\right) =
\sum_{j=1}^{2d} \left(\left(K_j^C \right)^\dagger K_j^C \braket{j}{l} \right) -
\braket{m_C}{l} \Id_{\XX}.
\end{equation}
\end{lemma}

\begin{proof}
First we compute $\braket{m_C}{l}$. We get
\begin{equation}
\braket{m_C}{l} = \sum_{i=1}^{2d} \Tr\left(K_i^C\rho_\infty^C \left( K_i^C 
\right)^\dagger\braket{i}{l}\right).
\end{equation}
Next, we move all the terms to one side of the equation and write all terms 
under the trace
\begin{equation}
\sum_{i=1}^{2d} \Tr\left(K_i^C\rho_{\infty} 
\left(K_i^C\right)^\dagger\braket{i}{l} - \frac{1}{2d}
\braket{m_C}{l} \rho_\infty^C\Id_{\XX} \right)=0,
\end{equation}
where we multiplied $\braket{m_C}{l}$ by $\rho_\infty^C \1_{\XX}$.
Finally, we use the fact that trace is cyclic and linear and get:
\begin{equation}\label{eq:lemmaL}
\Tr \rho_\infty^C \left(\sum_{i=1}^{2d} \left(K_i^C\right)^\dagger 
K_i^C\braket{i}{l} 
-\frac{1}{2d} \braket{m_C}{l}\Id_{\XX} \right)=0.
\end{equation}
Thus we obtain that the term under the bracket in Eq. (\ref{eq:lemmaL}) is 
orthogonal to $\rho_\infty^C$ and as it is the only invariant state of $\Phi^C$ 
we get that $\Ker(\1_{L(\XX)}-\Phi^C)=\rho_\infty^C$. Then, from Lemma 
\ref{lem:ker-im-sum}, the states orthogonal to the kernel are in the image of 
the conjugated superoperator, hence we get:
\begin{equation}
\sum_{i=1}^{2d} \left(K_i^C\right)^\dagger K_i^C \braket{i}{l} - \frac{1}{2d}
\braket{m_C}{l} \Id_\XX \in 
\Ker\left(\1_{L(\XX)} - 
\Phi^C\right)^\perp=\Image\left(\Id_{L(\XX)}-\left(\Phi^C\right)^\dagger\right).
\end{equation}
Hence, we have shown that $L_l^C$ exists.
\end{proof}

\begin{lemma}\label{lem:f}
For each class $C$ and a vector $l \in \mathbb{R}^d$ a function
\begin{equation}
f_C: \Omega(\XX) \times \mathbb{R}^d \rightarrow \mathbb{R},
\end{equation}
given by the explicit formula
\begin{equation}
f_C(\rho, i)= \Tr(\rho L_l^C) + \braket{i}{l}\label{eq:fc},
\end{equation}
satisfies
\begin{equation}
(1-P_C)f_C(\rho, i)=\braket{i}{l} - \braket{m_C}{l},
\end{equation}
where $P_C$ is given by Eq.~\eqref{eq:pc}.
\end{lemma}
\begin{proof}
We apply the $P_C$ operator as defined in Eq.~\ref{eq:pc}.
Let us note that 
\begin{equation}
(P_C f_C)(\rho, i) = \sum_{\rho', j}P_C[(\rho, i), 
(\rho', j)]f_C(\rho', j).\label{eq:action-of-P}
\end{equation}
Applying the definition of $P_C$ to \eqref{eq:fc} we get:
\begin{equation}
\begin{split}
(1 - P_C)f_C(\rho, i) & = \Tr(\rho L_l^C) + \braket{i}{l} - 
\left[\Tr\left(\sum_{j=1}^{2d} K_j^C\rho \left(K_j^C\right)^\dagger L_l^C 
\right) \right. \\
& + \left. \sum_{j=1}^{2d} \Tr \left(K_j^C\rho 
\left(K_j^C\right)^\dagger\right)\braket{j}{l}\right].
\end{split}
\end{equation}
Now, using Lemma~\ref{lem:Lexists} we get
\begin{equation}
\Tr \rho \left[
\left(\Id_{L(\XX)}-\left(\Phi^C\right)^\dagger\right)\left(L_l^C\right)-
\sum_{j=1}^{2d} \left(K_j^C\right)^\dagger K_j^C \braket{j}{l} \right] 
+\braket{i}{l}
=\braket{i}{l} - \braket{m_C}{l}.
\end{equation}
which completes the proof.
\end{proof}

\begin{proof}[Proof of Theorem~\ref{th:irreducible}]
For a random variable $X_n$ we expand the formula 
$F_l=\braket{X_n}{l}-n\braket{m}{l}$:
\begin{equation}
F_l =\braket{X_n}{l} - 
n\braket{m}{l}=\braket{X_0}{l}+\sum_{k=1}^n(\bra{X_k}-\bra{X_{k-1}})-\bra{m})\ket{l}.
\end{equation}
Recall that $\sum_{C \in \Gamma} p_C=1$, $m=\sum_{C \in \Gamma} p_C m_C$ and we 
denote $\bra{X_k} - \bra{X_{k-1}} = \bra{\Delta X_k}$, we get:
\begin{equation}
F_l =\braket{X_0}{l}+\sum_{k=1}^n\sum_{C \in \Gamma} p_C(\bra{\Delta X_k}- 
\bra{m_C})\ket{l}.
\end{equation}
From Lemma~\ref{lem:f} we get 
$(\ket{X}-\ket{m_C})\ket{l}=(1-P_C)f_C(\rho, \ket{X})$ for some $\rho 
\in \Omega(\XX)$, hence:
\begin{equation}
\begin{split}
F_l & =\braket{X_0}{l} + \sum_{k=1}^n\sum_{C \in \Gamma} p_C (1-P_C) 
f_C (\rho_k, \ket{\Delta X_k})= \\
& = \braket{X_0}{l} + \sum_{k=1}^n\sum_{C \in \Gamma} p_C (f_C(\rho_k, 
\ket{\Delta X_k}) - P_Cf_C(\rho_k, \ket{\Delta X_k})).
\end{split}
\end{equation}
After rearranging the sum in the formula for $F_l$ we get:
\begin{equation}
\begin{split}
F_l = &\braket{X_0}{l} + \sum_{k=2}^n\sum_{C \in \Gamma} [p_C(f_C(\rho_k, 
\ket{\Delta X_k})-P_C f_C(\rho_{k-1}, \ket{\Delta X_{k-1}}))] + \\ 
&+ \sum_{C \in \Gamma} p_C f_C(\rho_1, \ket{\Delta X_1})-\sum_{C \in \Gamma} 
p_C P_Cf_C(\rho_n, \ket{\Delta X_n})=M_n + R_n.
\end{split}
\end{equation}
Now we consider $M_n$ and $R_n$ separately. First we discuss $M_n$:
\begin{equation}
M_n = \sum_{C \in \Gamma} \sum_{k=2}^{n} p_C(f_C(\rho_k, \ket{\Delta X_k}) - 
P_Cf_C(\rho_{k-1}, \ket{\Delta X_{k-1}})).
\end{equation}
We notice that $M_n$ is a centered martingale i.e.
\begin{equation}
\mathbb{E}[\Delta M_n|\F_{n-1}] = 0,
\end{equation}
where $\Delta M_n = M_n - M_{n-1}$ and $\F$ denotes filtering for stochastic 
process $M_n$~\cite{brown1971, hall1980martingale}. This follows from the 
action of $P_C$ stated in eq. (\ref{eq:action-of-P}). As $P_C$ is a transition 
operator for the corresponding Markov chain, the value of $P_Cf_C$ for 
step $k-1$ is exactly the expectation value of $f_C$ at the next step
\begin{equation}
\E[ f_C(\rho_k, \ket{\Delta X_k})|\F_{k-1}] = P_Cf_C(\rho_{k-1}, \ket{\Delta 
X_{k-1}}).
\end{equation}
Additionally $|\Delta 
M_n|$ is bounded  from above i.e. 
$|\Delta M_n| < M_{\max}$ as $\Delta M_n$ includes terms corresponding to one 
step of the walk.

In the case of $R_n$ we have:
\begin{equation}
R_n=\braket{X_0}{l} + \sum_{C \in \Gamma} p_C f_C(\rho_1, \ket{\Delta X_1}) - 
\sum_{C \in \Gamma} p_C P_C f_C(\rho_n, \ket{\Delta X_n}).
\end{equation}
From the definition of $f_C$ we notice that $R_n$ is bounded as the first two 
terms are constant and the last one $P_Cf_C(\rho, \ket{\Delta X_n})=\Tr(\rho 
L_l^C)+\braket{\Delta X_n}{l}$ is clearly bounded, hence$|R_n| < R_{\max}$ and 
$R_n$ does not influence the asymptotic behavior.

Now it suffices to show that the following two equalities hold (for proof see 
Theorem 3.2 and Corollary 3.1 in~\cite{hall1980martingale}):
\begin{equation}
\lim_{n\rightarrow\infty} \frac{1}{n}\sum_{k=1}^n \E[(\Delta 
M_k)^2{1}_{|\Delta M_k|\ge \epsilon \sqrt{n}}| \F_{k-1}] = 0
\label{eq:matingale-zero},
\end{equation}
and 
\begin{equation}
\lim_{n\rightarrow\infty} \frac{1}{n} \sum_{k=1}^n \E[(\Delta M_k)^2| \F_{k-1}] 
= \sigma^2
\label{eq:matingale-sigma},
\end{equation}
to obtain that $M_n/\sqrt{n}$ converges in distribution to $\N(0, \sigma^2)$, 
where
\begin{equation}
1_{|\Delta M_k|\ge \epsilon \sqrt{n}}=\left\{
\begin{matrix}
1, {|\Delta M_k| \ge \epsilon \sqrt{n}},\\
0, {|\Delta M_k| <   \epsilon \sqrt{n}},
\end{matrix}
\right.
\end{equation}
introduces restricted expectation values.

We prove Eq.~\eqref{eq:matingale-zero} using the fact that $|\Delta M_k|$ is 
bounded, hence the sum in Eq.~\eqref{eq:matingale-zero} terminates for $n>N$, 
for some $N 
\in \mathbb{N}$.

In order to prove the equality in Eq.~\eqref{eq:matingale-sigma}  we expand 
$(\Delta M_k)^2$:
\begin{equation}
\begin{split}
(\Delta M_k)^2 & = \left( \sum_{C \in \Gamma} 
p_C (\Tr\rho_kL_l^C - \Tr \rho_{k-1}L_
(\bra{\Delta X_k}-\bra{m_C})\ket{l} ) \right)^2= \\
& \left(\sum_{C \in \Gamma} p_C \Delta M_k^C \right)^2= \sum_{C,C'\in \Gamma} 
p_C p_{C'} \Delta M_k^C \Delta M_k^{C'},
\end{split}
\end{equation}
where $\Delta M_k^C=( \Tr(\rho_kL_l^C)-\Tr(\rho_{k-1}L_l^C)+(\bra{\Delta 
X_k}-\bra{m_C})\ket{l})$. Next, we expand the product
\begin{equation}
\begin{split}
\Delta M_k^C \Delta M_k^{C'}&= 
[\Tr(\rho_kL_l^C)-\Tr(\rho_{k-1}L_l^C)+(\bra{\Delta 
X_k}-\bra{m_C})\ket{l}] \times \\
&\times [\Tr(\rho_kL_l^{C'})-\Tr(\rho_{k-1}L_l^{C'})+(\bra{\Delta 
X_k}-\bra{m_C'})\ket{l}].
\end{split}
\end{equation}
We divide this expression into three terms $\Delta M_k^{C'} \Delta M_k^C = 
T_{C,C'}^{(1,k)} + T_{C,C'}^{(2,k)} + T_{C,C'}^{(3,k)}$. Henceforth, we will 
drop indexes $C,C', k$ when unambiguous. The term $T^{(1)}$ is equal to:
\begin{equation}
T^{(1)}=\Tr(\rho_kL_l^{C'})\Tr(\rho_kL_l^{C})-\Tr(\rho_{k-1}L_l^{C'})\Tr(\rho_{k-1}L_l^C).
\end{equation}
We compute $\E(T^{(1)}|\F_{k-1})$ by adding the term $\pm \Tr(\rho_kL_l^{C'}) 
\Tr(\rho_kL_l^C)$
\begin{equation}
\begin{split}
\E[T^{(1)}|\F_{k-1}] & =\E(\Tr(\rho_{k}L_l^{C'})\Tr(\rho_{k}L_l^C)|\F_{k-1}) 
-\Tr(\rho_{k}L_l^{C'})\Tr(\rho_{k}L_l^C)+ \\
& +\Tr(\rho_{k} L_l^{C'}) \Tr(\rho_{k} L_l^C)- \Tr(\rho_{k-1} L_l^{C'}) 
\Tr(\rho_{k-1} L_l^C),
\end{split}
\end{equation}
we obtain a sum of two terms that can be interpreted as an increment part of a 
martingale and an increment part of a sum respectively. Thus after a summation 
over $k$ both terms 
are bounded and we get the equality
\begin{equation}
\lim_{n\rightarrow\infty} \frac{1}{n} \sum_{k=1}^n \sum_{C,C' \in \Gamma} p_C 
p_{C'} 
\E[T^{(1,k)}_{C,C'} | \F_{k-1}]=0.
\end{equation}
The term $T^{(2)}$ is given by:
\begin{equation}
T^{(2)} = -\Tr(\rho_{k-1} L_l^{C'})\Delta M_k^C -\Tr(\rho_{k-1} L_l^C)\Delta 
M_k^{C'}.
\end{equation}
We note that $\E(\Delta M_k|\F_{k-1})=0$. Thus after summation over $C$ and 
${C'}$ we get the the expectation value of the whole term $T^{(2,k)}_{C,C'}$:
\begin{equation}
\lim_{n\rightarrow\infty} \frac{1}{n} \sum_{k=1}^n \sum_{C,C' \in \Gamma } p_C 
p_{C'} 
\E[T^{(2,k)}_{C,C'}|\F_{k-1}] = 0.
\end{equation}

We will calculate the term $T^{(3)}$ using the definition of the expectation 
value. We write the probability of $\ket{\Delta X}$ being equal to $\ket{j}$ 
and $\rho_k$ being $K_j \rho_{k-1} K_j^\dagger / \Tr(A_j \rho_{k-1} 
A_j^\dagger)$ as $\Tr(K_j \rho_{k-1} K_j^\dagger)$. This can be expressed in a 
nice trace form:
\begin{equation}
\begin{split}
\E[ T^{(3)}_{C, {C'}}|\F_{k-1}] & =  \E[(\bra{\Delta 
X_k}-\bra{m_C})\ket{l}(\bra{\Delta 
X_k}-\bra{m_{C'}})\ket{l} +
\\&
+ \Tr(\rho_kL_l^{C'})(\braket{\Delta X_k}{l}-\braket{m_C}{l})+
\\&
+ \Tr(\rho_kL_l^C)(\braket{\Delta X_k}{l}-\braket{m_{C'}}{l})|\F_{k-1}]=
\\&
= \sum_{i=1}^{2d} \Tr(K_i^{c(k-1)}\rho_{k-1} {K_i^{c(k-1)\dagger}}) \times \\&
\times [(\bra{i}-\bra{m_C}) \ket{l} (\bra{i} - \bra{m_{C'}}) \ket{l}+
\\&
+ \Tr(K_i^{c(k-1)}\rho_{k-1} K_i^{c(k-1)\dagger} L_l^{C'}) (\braket{i}{l} - 
\braket{m_C}{l})\times
\\&
\times \Tr(K_i^{c(k-1)}\rho_{k-1} K_i^{c(k-1)\dagger} L_l^{C}) (\braket{i}{l} - 
\braket{m_{C'}}{l})],
\end{split}
\end{equation}
where $c(k-1)$ is the class of $X_{k-1}$. Thus we can define 
$\Xi_{C,C'}^{c(k-1)}$ so that
\begin{equation}
\E[ T^{(3)}_{C,C'}|\F_{k-1}] =\Tr(\rho_{k-1}\Xi_{C,C'}^{c(k-1)}).
\end{equation}
After summation over $C$ and $C'$ the value is equal to:
\begin{equation}
\lim_{n\rightarrow\infty} \frac{1}{n} 
\sum_{k=1}^n \sum_{C,C' \in \Gamma}p_C p_{C'} \E[T^{(3,k)}_{C,C'}|\F_{k-1}]
=\lim_{n\rightarrow\infty} \frac{1}{n} \sum_{k=1}^n \Tr(\rho_{k-1}\Xi^{c(k-1)}),
\end{equation}
where $\Xi^{c(k-1)}=\sum_{C,C'\in \Gamma}p_C p_{C'} \Xi_{C,C'}^{c(k-1)}$. By 
the 
ergodic theorem (Th. 4.2 in \cite{attal2012central}) this 
converges to:
\begin{equation}
\lim_{n\rightarrow\infty} \frac{1}{n} \sum_{k=1}^n \Tr(\rho_{k-1}\Xi^{c(k-1)})  
= \Tr(\rho_{\infty} \Xi) = \sigma^2_l,
\end{equation}
with $\Xi=\sum_{c} p_c \Xi^{c}$.

Finally, after summing of all of the terms we get:
\begin{equation}
\lim_{n\rightarrow\infty} \frac{1}{n} \sum_k \E[(\Delta M_k)^2|\F_{k-1}] = 
\sigma_l^2,
\end{equation}
which completes the proof.
\end{proof}
\subsubsection{Example}\label{sec:example-irreducible}
As an example of a walk consistent with description in 
Section~\ref{sec:irreducible-theorem} we consider a walk with the same vertex 
types as in the reducible case, that is:
\begin{equation}
\begin{split}
A_U(X) = & \alpha \ket{0}\bra{0}X\ket{0}\bra{0} + (1 - \alpha) \ket{1}\bra{0} 
X \ket{0}\bra{1}, \\
A_R(X) = & \frac12 \ket{1}\bra{1}X\ket{1}\bra{1} + \frac12 \ket{3}\bra{1} X 
\ket{1}\bra{3} \\
A_D(X) = & \alpha \ket{3}\bra{2}X\ket{3}\bra{2} + (1 - \alpha) \ket{2}\bra{2} 
X \ket{2}\bra{2}, \\
A_L(X) = & \frac12 \ket{3}\bra{3}X\ket{3}\bra{3} + \frac12 \ket{0}\bra{3} 
X \ket{3}\bra{0}.
\end{split}
\end{equation}
and
\begin{equation}
\begin{split}
B_U(X) = & \alpha \ket{1}\bra{0}X\ket{0}\bra{1} + (1 - \alpha) \ket{3}\bra{0} 
X \ket{0}\bra{3}, \\
B_R(X) = & \frac12 \ket{0}\bra{1}X\ket{1}\bra{0} + \frac12 \ket{2}\bra{1} X 
\ket{1}\bra{2}, \\
B_D(X) = & \alpha \ket{1}\bra{2}X\ket{2}\bra{1} + (1 - \alpha) \ket{3}\bra{2} 
X \ket{2}\bra{3}, \\
B_L(X) = & \frac12 \ket{0}\bra{3}X\ket{3}\bra{0} + \frac12 \ket{2}\bra{3} 
X \ket{3}\bra{2}.
\end{split}
\end{equation}
Although, in this case we assign the type to a vertex randomly with a uniform 
distribution.

The channels formed from Kraus operators $A_x$ and $B_x$ where $x \in {U, R, L, 
D}$ both have a unique invariant state. The behavior 
of the network is presented in the Fig.~\ref{fig:2d-irreducible-walk}. We 
obtain a similar behavior as in the reducible case, although the convergence to 
a Gaussian distribution is slower.
\begin{figure}[!h]
\centering
\subfloat[\label{fig:irreducible-a}]{\includegraphics[width=0.49\textwidth]{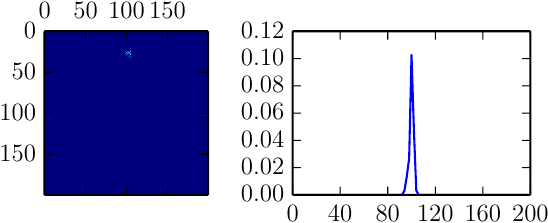}}
\subfloat[\label{fig:irreducible-b}]{\includegraphics[width=0.49\textwidth]{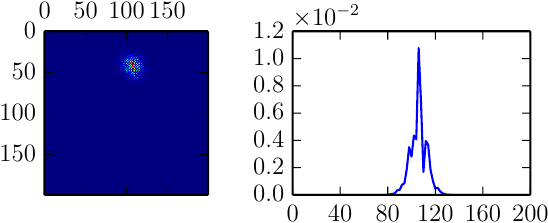}}\\
\subfloat[\label{fig:irreducible-c}]{\includegraphics[width=0.49\textwidth]{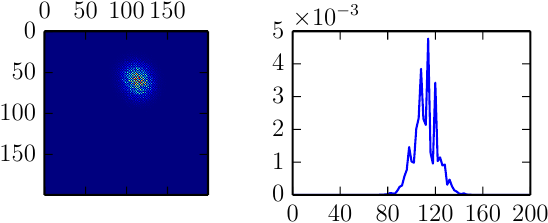}}
\subfloat[\label{fig:irreducible-d}]{\includegraphics[width=0.49\textwidth]{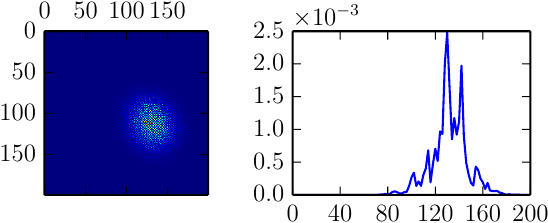}}
\caption{An example of realization of OQW with a random uniform distribution of 
vertex types. The figures show the distribution of the walk and cross section 
through the center for various time steps:
panel~\protect\subref{fig:reducible-a}  $n=10$,
panel~\protect\subref{fig:reducible-b}  $n=100$,
panel~\protect\subref{fig:reducible-c}  $n=200$,
panel~\protect\subref{fig:reducible-d}  $n=500$.}\label{fig:2d-irreducible-walk}
\end{figure}

\section{Conclusions}
The aim of this paper was to provide formulas describing the behavior of the
open quantum walk in the asymptotic limit. We described two cases: networks
that are reducible to the 1-type case and networks with random, uniformly
distributed vertex types. This result allows one to analyze behavior of walks
with a more complex structure compared to the known results. We have illustrated
our claims with numerical examples that show possible applications and
correctness of our theorems. The networks are still restricted to vertices that
exhibits invariant states.

We provided examples showing that the theorems are valid in the case of a 2D
regular lattice with two vertex types. In Section~\ref{sec:example-reducible}
we shown application to the reducible case, when the assignment of vertex
types is regular and translation invariant. Next, in
Section~\ref{sec:example-irreducible} we turned to a random, uniformly
distributed assignment of vertex types.

These theorems can also be applied to the non-lattice graphs. Different types 
of vertices allow also to apply this in the case of graphs with non-constant 
degrees. This may be very useful in modeling complex structures, especially of 
regular definition as in the case of Apollonian networks.

These possibilities are important as open quantum walks with different vertex 
classes have application in quantum biology and dissipative quantum computing.

\section*{Acknowledgments}
We would like to thank Hanna Wojew{\'o}dka for fruitful discussions and a 
critical reading of our manuscript.

Work by {\L}P was supported by the Polish Ministry of Science and Higher
Education under the project number IP2012 051272. PS was supported by the Polish
Ministry of Science and Higher Education within ``Diamond Grant'' Programme
under the project number 0064/DIA/2013/42.
\bibliography{../manuscript}
\bibliographystyle{ieeetr}
\end{document}